\documentclass[11pt]{llncs}
\newenvironment{protocol}[1]{\begin{protocolth}\begingroup\rm {\it #1}}
{\endgroup\end{protocolth}\vspace{0.3cm}}
\newtheorem{protocolth}{Protocol}
\newtheorem{cor}{Corollary}
\newtheorem{istmt}{Important Statement}

\usepackage{amssymb}
\usepackage{nopageno}
\usepackage{setspace}
\usepackage{amsmath}
\usepackage{amsbsy}
\usepackage{enumerate}
\usepackage{url}
\usepackage{graphicx}
\usepackage{rotating}
\usepackage{pdfpages}
\usepackage{footmisc}
\usepackage{pdflscape}
\usepackage{float}

\begin{document}
\title{Making Code Voting Secure against Insider Threats using Unconditionally Secure MIX Schemes and Human PSMT Protocols}
\author{\hfill\\[-1.5cm]}
\institute{}
\author{Yvo Desmedt\inst{1,2}
     \thanks{A part of this work was done while being, part time,
      at RCIS/AIST, Japan.}
    \and Stelios Erotokritou\inst{1,3}}

     \institute{Department of Computer Science, The University of Texas at Dallas, US
     \and
     Department of Computer Science,University College London, UK\\ \email{\{y.desmedt,s.erotokritou\}@cs.ucl.ac.uk}
     \and
     CaSToRC, The Cyprus Institute, Nicosia, Cyprus}

\maketitle
\begin{abstract}
Code voting was introduced by Chaum as a solution for using a possibly
infected-by-malware device to cast a vote in an electronic voting application.
Chaum's work on code voting assumed voting codes are physically delivered to
voters using the mail system, implicitly requiring to trust the mail system.
This is not necessarily a valid assumption to make - especially if the mail system cannot be trusted.
When conspiring with the recipient of the cast ballots, privacy is broken.

It is clear to the public that when it comes to privacy, computers and ``secure'' communication over the Internet cannot fully be trusted. This emphasizes the importance of using: (1) Unconditional security for secure network communication. (2) Reduce reliance on untrusted computers.

In this paper we explore how to remove the mail system trust assumption in code voting.
We use PSMT protocols (SCN 2012) where with the help of visual aids, humans can carry out $\mod 10$ addition correctly with a 99\% degree of accuracy.
We introduce an unconditionally secure MIX based on the combinatorics of set systems.

Given that end users of our proposed voting scheme construction are humans we \emph{cannot use} classical Secure Multi Party Computation protocols.

Our solutions are for both single and multi-seat elections achieving:
\begin{enumerate}[i)]
  \item An anonymous and perfectly secure communication network secure against a $t$-bounded passive adversary used to deliver voting,
  \item The end step of the protocol can be handled by a human to evade the threat of malware.
\end{enumerate}
We do not focus on active adversaries.
\end{abstract}
\section{Introduction}\label{sec:intro}
\setcounter{page}{1}
    Electronic voting over the Internet enables to cast votes from an Internet-connected device from \emph{any} physical Internet accessible location compared to booth based electronic voting systems developed by the cryptographic community~\cite{scantegrityJournal,NHR2004}. Internet voting does not require voters to be physically present at a polling station.

    Even though secure Internet voting is in its infancy, many countries and organizations are considering adoption or have already done so such as Estonia~\cite{Maaten04} and Switzerland~\cite{switzEvoting}.
    In Estonia, participation increased by 17\%~\cite{reason1}.
    Similarly, after IACR used the Helios Internet voting system~\cite{heliosVotingSystemsSS} which allowed its member's who are based in different \emph{geographical} locations to cast their secure vote online, voting increased from 20\% to around 30\%-40\%.

    Experts agree that achieving secure Internet voting will be even more difficult than booth-based electronic voting~\cite{rivestRubin}. For example, the 2003 CRA Grand Research Challenges Workshop on Information Security~\cite{grandChallenge} ranked secure Internet voting as one of the most challenging open problems in information security. These issues were put in the spotlight at the 2013 RSA Conference panel~\cite{RSAPanel} and by Rivest in~\cite{rivestTalk}.
    The difficulties lie in the fact that computational devices are vulnerable to security attacks and are easy to hack. Although SSL uses cryptography, modern browsers are vulnerable to attacks such as click-jacking, cross-site scripting, and man-in-the-browser attacks - as demonstrated against Helios 2.0 in~\cite{EstehghariDesmedt10}.

    Given that the computer of a voter can easily be hacked, in 2001 Chaum proposed a breakthrough solution called ``code voting''~\cite{Chaum01} where one can use a possibly hacked computer to perform a secure operation.
    In code voting, a voter receives through the \emph{postal mail} a long enough unique code for every candidate. To vote, voters would just enter the code corresponding to the candidate of their choice.

    Chaum's approach to code voting assumes the postal mail to be secure from a reliability and privacy viewpoint.
    This is not a valid assumption to make.
    Indeed, a collaboration of the postal service with the returning officer\footnote{A returning officer oversees elections in one or more constituencies~\cite{returningOfficer}.} may allow for the anonymity of all votes to be broken by divulging the identity of voters to whom specific voting codes were delivered.
    Another problem~\footnote{\label{first_footnote}Since we focus on a passive adversary, our paper does not address this attack.} is that if one knows who is likely not to vote, Chaum's scheme is not very secure against ballot stuffing.
    Furthermore, if malicious postmen do not deliver voting codes, this prevents
    voters from casting their votes~\footref{first_footnote}. If the election is tight and the number of undelivered ballots is high, this could undermine the reliability and trustworthiness of code voting through the postal service.
    So, one question we address is \emph{how we can make Chaum's code voting secure against $t$ passive insiders}.

    Obviously we need to maintain the anonymity of voters.
    One way to achieve anonymity is through the use of MIX-networks. These were first introduced by Chaum in~\cite{Chaum81} and are used in electronic voting. MIX-networks allow senders to input a number of (usually encrypted) messages to a MIX-network which then outputs and delivers each message to all recipients without the receiver being able to identify the sender. Various ways with which MIX-networks are constructed are described in Section~\ref{sec:prevWorks}. The main issue with such constructions is that they are based on tools based on computational assumptions which when used within the context of an electronic voting scheme only allows for \emph{conditional security} thus conditional privacy and conditional anonymity to be achieved.

    Note that no conditional secure cryptosystem designed so far has withstood cryptanalysis for more than 300 years. Quantum computers will undermine computational voting schemes cryptographers have proposed, in particular these based on ElGamal. For many goals, unconditionally secure solutions have already been proposed, e.g., since 1988~\cite{BenOrGoldwasser88,ChaumCrepeauDamgard88} we have unconditionally secure multiparty computation. This is a further motivation for proposing an unconditionally secure voting scheme in which $t$ insiders can be corrupted. 
    Due to the revelations by Snowden~\cite{Snowden0}, some have questioned the security of the NIST standards~\cite{Snowden2}. So, one can wonder whether we want to promote voting systems which might be broken, if not now, then in the future. The importance of requiring unconditional vote security is further highlighted with the following example:
    \begin{quote}
    In 2020 Alice turns 18 and votes using a popular ElGamal based electronic voting scheme. 50 years later, Alice is a candidate for president of the USA. Imagine that in 2070 USA politics is going through a new McCarthy~\cite{mccarthy} witch hunt. Unfortunately for Alice, ElGamal security has since been broken. The newspapers find that Alice voted for the what is then considered the ``\emph{wrong}'' party!
    \end{quote}
    In this paper we focus on \emph{unconditional security} proposing alternative MIX constructions (using set systems and shares of messages), to generate the \emph{correctness of the vote unconditionally}.
    To counter technological threats and the possible influence of elections by foreign governments (where hardware are manufactured), our proposed Internet code voting solution uses the concept of \emph{diversity}, first described in~\cite{diverseForrestSA97}. So, we employ a \emph{diversity of computing systems} to achieve security in our proposed solution.
    Using diversity of network paths we also ensure that any $t$-bounded adversarial presence is unable to break the privacy of any votes. We consider the $t$-bounded computationally unlimited adversary to be capable of taking control of any node between the vote authority and the voters which includes nodes in the MIX-network, nodes in the communication network or voters computational devices (through malware). It should be noted that we do not consider the human voters to be corruptible.

    The main part of our work assumes a passive adversary which can only observe but cannot cause deviation of protocol execution in any way. We also assumes that the adversary cannot look at the information on the whole network but only inside $t$ nodes.
    Our solution considering an active adversary will be presented in a future full version of this paper.

    Considering a $t$-bounded adversary we emphasise the following:
    \begin{istmt}\label{istmt:statementFranklin}
            As shown in~\cite{FranklinY04}, when the number of corrupted nodes is at most $t$, the minimum number of disjoint paths to allow for private communication between a sender and a receiver is $t+1$.
    \end{istmt}
    \begin{cor}\label{cor:corStatementFranklin}
    Because of the above, voters will have to use a number of computing devices to securely receive (or dually send) their voting codes.
    \end{cor}
    The impact of Corollary~\ref{cor:corStatementFranklin} is not as bad as it might initially seem. \emph{Nowadays, many people in developed countries can have effortless access to more than one device} such as PC's, laptops, smartphones and tablets.  Such devices can include those they own or can access through friends and relatives or through public access (such as a library). Furthermore, each of these devices can be connected to a communication network in a different manner (Internet or  cellular) which could be serviced by different providers. These devices may run different operating systems (e.g. Windows, IoS, Android) thus a threat to one device may not necessarily constitute a threat to another - even with the same user.

    Similar to the work of~\cite{BlockiBD14,HopperB01} which considers security protocols as used by humans who can execute them without relying on a fully trusted computer we do the same in this paper in the context of Internet voting.

    Motivated by all the above, we propose an unconditional Internet code voting protocol which is secure against the possible presence of an adversary and malware in the network and on voter's devices respectively.
    We present solutions for single seat and multi-seat elections both of which are designed to be user friendly - so that human voters can use it correctly with high accuracy\footnote{It should be noted that the main goal of our work is Internet code-voting secure against $t$ insiders. The work of~\cite{financialCrypto} is independent and their MIX servers are different using a homomorphic, unconditionally hiding
    commitment scheme to encrypt audit information and achieve unconditional security. Furthermore, their solution assumes the use of two mix-networks one of which is private and thus cannot be corrupted by the adversary. Our solution does not make this assumption and instead counters the threat of the adversary presence for protocol correctness accordingly. However, due to the possible presence of malware the only way we know how to achieve this, is using unconditionally secure techniques achieved through the use of cover designs. Additionally we use results from previous work~\cite{humanPSMTSCN} which allows for humans to privately and reliably receive and decode messages, something achieved with unconditional security.}.
    In EVOTE2014~\cite{RR14b} the authors addressed a very similar problem as our current work. However, their solution achieves conditional security which could be broken in the future against a computationally unlimited adversary. Furthermore, the authors consider the adversary to be present in the MIX network only and do not take into account the possible presence of malware upon the tablets with which voters will use to cast their votes. Passive malware could possibly identify to an adversary how someone voted, whereas active malware could alter the way someone votes - thus rigging the result of an election.

    When combined with~\cite{humanPSMTSCN}, one can view our proposed method for delivering codes to voters as a distributed implementation of a one-time-pad-secured communication channel for votes. Because of this, our solution can also be used for other established code voting schemes as it is a way of removing the use of a possibly untrusted mail system and transmitting the voting codes securely, reliably and anonymously to voters.

    The text is organized as follows. Background and relevant previous work are presented in Section~\ref{sec:backGround}. In Section~\ref{sec:highLevelDescr} a high level description of the protocol is given and we identify the required cryptographic tools. In Section~\ref{sec:didacticTransmitReply} we provide a simplified version of the MIX private and anonymous communication protocol. This is used in Section~\ref{sec:privAnonCommunAbelian} in a more efficient manner where we present private and anonymous communication protocols for the transmission of voting codes to voters for single seat and multi seat elections. In Section~\ref{sec:evotingProt} the electronic code voting protocol is presented and the security proof of the protocol is also given.
\section{Background and Previous Work}\label{sec:backGround}
\subsection{Previous Work}\label{sec:prevWorks}
    This section describes previous work related to various aspects to be presented in this paper.

    MIX-networks can be constructed using a shuffle (permutation). One way of achieving this~\cite{KhazaeiMW12,kishnaMix} is by using approaches which are based on zero-knowledge arguments~\cite{Furukawa05,Wikstrom02a}. In~\cite{DesmedtKurosawa00Eurocrypt} the use of zero-knowledge was avoided. MIX-networks based on zero-knowledge arguments can be used in electronic voting protocols - as has been proposed in recent publications~\cite{Groth09CRYPTO,GrothI08EUROCRYPT}. Earlier work~\cite{SakoK94} similarly used shuffles in electronic voting based on zero-knowledge proofs. Other work on MIX-networks includes the work of Abe in~\cite{Abe98Eurocrypt}.

    Such constructions are based on computational assumptions which only allow for \emph{conditional security}. The work we present is based on the stronger model of \emph{unconditional security}.

    Anonymity in practice is difficult to achieve.
    One proposed implementation was that of~\cite{KattiCohenKatabi07} but it was shown to be insecure in~\cite{hachItOut}.

    A voting scheme similar to the one we propose which achieves information theoretic security and requires the voter to carry out modular addition is that presented in~\cite{MoranN10SplitTrust}. Contrary to the voting scheme proposed in this paper, the work of~\cite{MoranN10SplitTrust} is \emph{not} an Internet voting scheme as it requires voters to cast their votes at a polling station.

    The work of~\cite{CramerFSY96linear} describes an election scheme which requires computational modular exponentiation operations to be carried out by voters. These operations require software or hardware. Furthermore, public key-cryptography is used, meaning that the security properties achieved are computational and not information theoretic - as achieved in our scheme.
\subsection{Message Transmission Security Properties}\label{sec:securityProperties}
    Below we define message transmission security properties which will be required throughout the text. For formal definitions, see~\cite{DolevDworkWaartsYung93}.
    In our setting we have a single receiver $S$ connected to $m$ number of senders ($r_1,\cdots,r_m$) over a possibly corrupt underlying network.

    \textbf{(Perfectly) Correct} - When the receiver accepts a message, it was sent by a sender $S$.

    \textbf{(Perfectly) Reliable} - When a sender $S$ transmits a message, this message will be received by the receiver with probability 1.

    \textbf{(Perfectly) Private} - Only the designated receiver(s) can read a message transmitted by $S$. i.e., for any coalition of $t$ parties, their probability of correctly determining a message is the same whether the coalition is given their transmission view or not.

    \textbf{(Perfect) Security} - Means perfect correctness, perfect reliability and perfect privacy.

    \textbf{(Perfectly) Anonymous} - Considering the single receiver wants to receive $m$ different messages over the network so that each of $m$ number of senders transmitted one of these messages (and each message is transmitted and received only once), perfect anonymity is achieved when for any coalition of $t$ parties, their probability of correctly determining the sender of \emph{any} message is the same whether the coalition observes the transmission view or not.
    In the context of Internet voting, perfect anonymity is achieved when the voting protocol used does not facilitate any party involved in the voting process to correlate any cast vote to a specific voter with greater probability than any other.
\subsection{Existential Honesty}
    Some of our ideas use concepts of {\it existential honesty}, defined in~\cite{DesmedtKurosawa00Eurocrypt} as:
    \begin{quote}
        ``It is possible to divide the MIX servers into blocks, which guarantee that one block is free of dishonest MIX servers, assuming the number of dishonest MIX servers is bounded by $t$.''
    \end{quote}
    To achieve this,~\cite{DesmedtKurosawa00Eurocrypt} defined and used the following (see also~\cite[Section~2.3]{humanPSMTSCN} for an extensive description of set systems and how these relate to covering designs.):
    \begin{definition}[\cite{colbournHandbook}]\label{def:DesmedtKurosawaSetSystem}
    A set system is a pair $(X,{\cal B})$, where $X \triangleq \{1,$ $ 2, \ldots, m\}$ and $\cal{B}$ is a collection of blocks $B_i \subset X$ with $i = 1, 2, \ldots, b$.
    \end{definition}
    \begin{definition}[\cite{DesmedtKurosawa00Eurocrypt}]\label{def:DesmedtKurosawa}
    $(X,{\cal B})$ is an  $(m,b,t)$-verifiers set system if:
    \begin{enumerate}
    \item $|X|=m$,
    \item\label{cond:size-of-B} $|B_i| = t+1$ for $i = 1,2,\ldots,b$, and
    \item\label{cond:free-of-bad-ones} for any subset $F \subset X$ with $|F| \leq t$, there exists a $B_i \in {\cal B}$ such that $F \cap B_i = \emptyset$.
    \end{enumerate}
    \end{definition}

    We \emph{assume that private channels} connect MIX servers of corresponding blocks (i.e. when for block $B_k$, MIX server $MIX_{k,i}$ needs to communicate with MIX server $MIX_{k+1,j}$, where $1 \leq i,j \leq t+1$ and $k<b$, then there is a private channel). We also assume such channels between the receiver and $MIX_{1,i}$ and similarly, between $MIX_{b,i}$ and the sender.
\subsection{Human Perfectly Secure Message Transmission Protocols}\label{sec:otherSubprotsHumanPSMT}
    Perfectly secure message transmission (PSMT) protocols where the sender or receiver is a human were introduced in~\cite{humanPSMTSCN}.
    In such protocols it is assumed that the human receiver does \emph{not} have access to a trusted device since these may be faulty and/or infected with malware.
    Because the receiver is a human, such protocols aim to achieve perfectly secure message transmission (PSMT) in a computationally efficient and computationally simple manner. It is also important that the amount of information and operations the human receiver should process be kept to a minimum.

    Addition $\bmod  10$ was used by humans in these protocols~\cite{humanPSMTSCN} to reconstruct the secret message of the communication protocol from received shares through addition $\bmod  10$. The idea of using addition $\bmod  10$ for human computable functions was also used in~\cite{BlockiBD14} but within a different security context.
    By regarding in~\cite{humanPSMTSCN} $Z_{10}(+)$ as a subgroup of $S_{10}$ the operation became very reliable for humans to perform. Experiments have shown that given clear, correct and precise instructions, coupled with visual aids, allowed for the correct usage of these protocols by a very high percentage of human participants.
\subsection{Secure Multiparty Computation in Black-box Groups}\label{sec:otherSubprotsBlackBox}
    Black box multiparty computation protocols against a passive adversary for non-Abelian group have been presented in~\cite{CohenDIKMRR13} and in~\cite{DesmedtYao12} through the use of a $t$-reliable $n$-coloring admissible planar graph.
    These papers studied in particular the existence of secure $n$-party protocols to compute the $n$-product function $f_G(x_1,$ $\cdots,$ $x_n) := x_1 \cdot \ldots \cdot x_n$ where each participant is given the private input $x_i$ from some non-Abelian group $G$ where $n\geq 2t+1$. It was assumed that the parties are only allowed to perform black-box operations in the finite group $G$, i.e., the group operation $((x, y) \mapsto x \cdot y)$, the group inversion  $(x \mapsto x^{-1})$ and the uniformly random group sampling ($x \in_R G$).
\section{Secure Code Voting with Distributed Security}\label{sec:highLevelDescr}
    In this section we provide a high level description of the secure code voting protocol we will present in this paper. We assume the reader is familiar with Chaum's code voting scheme~\cite{Chaum01}.
\subsection{High Level Description}\label{sec:description}
    We call Code Generation Entity (CGE) the entity in the code voting protocol which is responsible for creating the codes with which voters will cast their votes.
    These codes are unique and are sent to the voters so that each of these codes is used only once for the \emph{whole} election. For single seat elections each voter receives as many codes as there are candidates. For multi-seat elections each voter receives a single permutation - which is a permutation of the alphabetical ordering of the candidates.
    After these codes pass through a MIX network, they will be sent to voters using perfectly secure message transmission, i.e. using secret sharing.
    Voters will receive each share using a different device, identify the shares which correspond to the candidate of their choice and reconstruct using human computation this voting code.
    To cast their vote, voters will send this code back to the CGE via the MIX servers, which perform inverse operations. For each of the received cast codes, the CGE will identify the candidate the code corresponds and will tally up the cast votes for each candidate.

    Our protocol does \emph{not} use the mail system for the delivery of voting codes to voters, but instead these are sent by the CGE to voters over a MIX network and using PSMT. Similarly, cast votes will be sent by voters to the CGE over a network as explained in Section~\ref{sec:sendinVotes}.
\subsection{Required Cryptographic Tools}\label{sec:requiredTools}
    The process should not facilitate the CGE (and indeed any $t$ other parties) \emph{should not} be able to identify that a specific voter (from the set of $v$ voters) cast a particular vote.
    Furthermore, \emph{a number of the underlying network nodes may be corrupt}. Even though secret sharing is used, any protocol should ensure that voting codes are not learned by any $t$ parties apart from voters themselves, otherwise anonymity of votes could be broken.

    Human perfectly secure message transmission protocols as presented in~\cite{humanPSMTSCN} are employed. We rely on the feasibility tests performed which confirm that humans can perform these basic operations.
    We use the secret sharing scheme friendly to humans as presented in~\cite[Section~2.2]{humanPSMTSCN} which guarantees perfect privacy unconditionally.
    Except for the voters computing the codes from the shares they receive, all other computations are carried out by computers, of which no more than $t$ of these are curious.
\section{Transmit and Reply Protocol}\label{sec:didacticTransmitReply}
    In this section we present the first of the required primitives - a perfectly private and perfectly anonymous network communication protocol. For didactic purposes, the simplest form of our proposed protocol will be presented - with more efficient constructions described later.

    Suppose that we have a single receiver and $v$ senders each of whom needs to receive a secret one time pad so as to sender a secret back to the receiver in an interactive anonymous way\footnote{The dual problem is that instead of having $v$ senders, we have $v$ receivers and one sender. Obviously a solution for the first provides a similar solution for the second.}.

    We assume the passive adversary controls at most $t$ MIX servers.
    As in Chaum's work~\cite{Chaum81} and most conditional MIX servers, each MIX server is 
    involved in one mixing in our protocol. $t+1$ blocks of MIX servers will be required - denoted as $B_1, \ldots, B_{t+1}$, with each block consisting of $t+1$ MIX servers and we use $B_k=\{MIX_{k,1},$ $MIX_{k,2},$ $\ldots,$ $MIX_{k,t+1}\}$ to identify MIX servers of the $k^{th}$ block and call $MIX_{k,1}$ ``$B_k$'s leader''.
\subsection{Protocol Main Idea}\label{sec:protMainIdeaDidactic}
    The receiver will share each of the $v$ one-time pads to transmit into $t+1$ shares using XOR.
    Each (of the $t+1$) share will be given to a corresponding MIX server (i.e. one of the $t+1$ servers) in the first block $B_1$.

    The shares of the $i^{th}$ one-time pad and those of the $j^{th}$ one-time pad might be transposed and will also be altered. To guarantee shares of the same pad stay together, the transpositions and alterations are chosen by the block leader.
    After the last MIX operation, the final block of MIX servers delivers the shares of the one time pad to the senders, with each sender reconstructing the received and altered one-time pad sent by the receiver.
    Each sender will then XOR the secret message to be sent to the receiver with the received altered one-time pad and send the result to the receiver over the MIX network. During this reverse transmission, the inverse alterations will be applied by each block leader.

    By XOR'ing the one time pad initially sent out by the receiver, the secret message sent by each sender can be obtained by the receiver.
\subsection{The MIX Communication Protocol - 1A: Receiver to Sender Transmission}
    We now present the steps in the MIX communication protocol for the transmission of the one-time pads from the receiver to the set of senders.
    \begin{protocol}{Private and Anonymous Communication Protocol}\label{prot:privAnonProtD}
    \begin{enumerate}[{\bf {Step }1}]
        \item Let $\pi^1_{i}$ be the $i^{th}$ one-time pad (where $1\leq i\leq v$). The receiver shares each $\pi^1_{i}$ into $t+1$ shares $\pi^1_{i,j} \in F_{2^l}$ using XOR (where $1\leq j\leq t+1$) and privately sends $\pi^1_{i,j}$ to the corresponding MIX $MIX_{1,j}$ in block $B_1$.
        \item\label{enum:dPRAMIX1}
            The \emph{leader} of $B_1$ (we call $MIX_{1,1}$) informs all others MIX servers in $B_1$ how they have to permute the $i$-index of all above $\pi^1_{i,j}$. This permutation is defined by $\rho_1\in_R S_v$.
        \item\label{enum:dPRAMIX12}
            On the $i$ indices all MIX servers in $B_1$ apply the permutation $\rho_1$. So, $\pi^{1}_{i,j}:=\pi^1_{\rho_1(i),j}$.
        \item
            The \emph{leader} of $B_1$ chooses $t+1$ random bit string modifiers $\omega^1_{i,j}\in_R F_{2^l}$ and privately sends $\omega^1_{i,j}$ to parties in $B_1$.
        \item
            For each $(i,j)$ the $t+1$ values $\pi^1_{i,j}$ are regarded as shares of $\pi^1_{i}$. Similarly, the $t+1$ values $\omega^1_{i,j}$ are regarded as shares of $\omega^1_{i}$.

        The MIX server in $B_1$ computes $\pi^2_{ij}=\omega^1_{ij} + \pi^1_{ij}$. $\pi^{2}_{i,j}$ are regarded as shares of $\pi^{2}$, the $\rho_1(i)$ permuted and modified one time pad.
        \item
        Steps 2-5 are repeated, incrementing by one the indices of $B_1$ and $B_2$ until the last block $B_{b}$ is reached.
         \item
         Shares held by MIX-servers of block $B_{t+1}$ are denoted as $\phi_{i,j}$. $MIX_{t+1,j}\in B_{t+1}$ then sends $\phi_{i,j}$ to the $i^{th}$ sender.
    \end{enumerate}
    \end{protocol}
\subsection{The MIX Communication Protocol - 1B: Sender to Receiver Transmission}
    Upon the end of the first phase, each sender reconstructs their respective altered one-time pad using XOR over all shares received.
    Using this altered one-time pad, a sender encrypts their secret using XOR.

    Senders then proceed to send their encrypted secret to the \emph{leader} of block $B_{t+1}$. These are then sent back to the receiver in much the same way as transmitted from receiver to sender, only this time, data are sent between \emph{leaders} of MIX blocks, the inverse permutations will be applied and all modifiers used will now have be invalidated. Thus the leaders of each block of MIX servers will use the inverse permutations $\rho_{b}^{-1}$ and invalidation of modifiers $-\omega_{i}^k$ (invalidating using XOR).

    The data that are sent back to the receiver correspond to the encrypted message transmitted by senders, and by applying XOR to this using the respective one-time pad, the secret message transmitted by senders can be obtained.

    It should be noted, that this anonymous and private communication protocol can be used for various practical applications. One such example is anonymous therapy sessions with extensions of the protocol also allowing for anonymous feedback.
\subsection{Security Proof}
    In this section we present the security proof for Protocol~\ref{prot:privAnonProtD}.
    \begin{theorem}\label{the:secProofAnonProtD}
    Protocol~\ref{prot:privAnonProtD} is a reliable, private and anonymous message transmission protocol.
    \end{theorem}
    \begin{proof}
    The protocol achieves perfect reliability due to the passive nature of the adversary.
    Perfect privacy is achieved as each one-time pad or encrypted message is ``shared'' over $t+1$ shares. As each MIX server is used only once and as the adversary can control at most $t$ MIX servers, secrecy of these transmitted data is retained.
    We now prove the perfect anonymity of the protocol - for simplicity of the proof we assume that there are only two messages (two one time pads).

    As $t+1$ blocks of MIX servers are used and each MIX server is used only once,
    there exists a block $B_i$ - $1 \leq i \leq b$, free from adversary controlled MIX servers. Because of this, the adversary is unable to learn the modifiers and permutation which are added and implemented respectively to the shares of the messages.

    Assuming the adversary is present in $B_{i+1}$ and absent from $B_{i}$, the view of the adversary of a share for both messages can be one of the following two possibilities:
    $(\omega_1^i+\pi^{i-1}_1, \omega_2^i+\pi^{i-1}_2),\qquad (\omega_2^i+\pi^{i-1}_2, \omega_1^i+\pi^{i-1}_1)$
    
    Obviously, the adversary cannot distinguish between the first and the second possibility as the modifiers and permutation used in block $B_i$ are random and not learned by the adversary.
    Indeed, there exists an ($\omega_1',\omega_2'$) such that
    ($\omega_2^i+\pi^{i-1}_2,\omega_1^i+\pi^{i-1}_1$)=($\omega_1'+\pi^{i-1}_1,\omega_2'+\pi^{i-1}_2$).
    So, the adversary cannot distinguish whether the messages have been interchanged or not.

    Without loss of generality, the proof can be extended to any number $v$ of messages.
    \makeatletter \let\@xpar=\par\penalty10000\hfill$\Box$\vspace{.5ex}\@xpar\penalty-8000\noindent \makeatother
    \end{proof}
\section{Reducing the Number of MIX Servers}\label{sec:privAnonCommunAbelian}
    In this section we improve on the ``Transmit and Reply Protocol'' presented in Section~\ref{sec:didacticTransmitReply} presenting a solution for the single seat election case where an Abelian group is used.

    Our solution uses Chaum's code voting and considers a single receiver (e.g., CGE) and $v$ human voters who each need to receive voting codes (one code per candidate) in a non-interactive anonymous way. We consider the CGE as the receiver and the human voters as the senders of the communication because at the end of the combined protocol, the human voters will send back to the CGE the voting code which corresponds to the candidate of their choice.
    We regard codes intended for the same receiver as a long string and the MIX servers MIX the strings (i.e. those intended for different receivers).

    A more efficient network of MIX servers is used as our solution is not confined to using each MIX server only once, thus the total number of MIX operations done is $b$. We denote the set of MIX servers by $X$ and assume we have an $(X,{\cal B})$ set system, which is an $(m,b,t)$-verifiers set system set system as defined in~\cite{DesmedtKurosawa00Eurocrypt}.
    We let $B_k=\{MIX_{k,1},$ $MIX_{k,2},\ldots,$ $MIX_{k,t+1}\}$ and call $MIX_{k,1}$ ``$B_k$'s leader''.

    The main idea of the protocol is very similar to the communication protocol of the previous section.
    This time, the receiver (e.g., CGE) will share each of the $v$ messages to transmit using an appropriate secret sharing scheme (and not using XOR). In a similar fashion, messages are permuted and altered as they are transmitted within the MIX network. After the last MIX operation, the final block of MIX servers delivers the shares of messages to the senders, with each sender reconstructing the secrets (voting codes) sent by the receiver. We will assume the transmission of the shares of these secrets uses the human friendly method presented in~\cite{humanPSMTSCN}. Similarly, since a code is only used once, it can be modified using addition over a finite Abelian group. To be compatible with~\cite{humanPSMTSCN} one such example is addition $\bmod 10$ over the group used. Senders will then transmit back to the receiver the voting code corresponding to their choice.
\subsection{Virtual Directed Acyclic Graphs}\label{sec:virtDirAcycGrph}
When an Abelian group is used and when blocks of the $(m,b,t)$-verifiers set system can share common MIX servers between them, we define the construction of a {\it virtual\/} vertex-labeled Directed Acyclic Graph (DAG). The set of vertices of the DAG is composed of parties participating in the protocol (which is similar to Protocol~\ref{prot:privAnonProt}), with the source of the graph being the receiver of the protocol and the sink being a sender.

The directed edges of the DAG identify the transmission of messages from one party to another \emph{amongst different levels} in the DAG. We define levels of the DAG as the receiver, a sender and the different blocks of MIX servers used. Considering block $B_i$ as a tuple (ordered set),
when $B_i$ is a block where $|B_i|=l$ and $b\in B_i$, at {\it location\/} $k$ in this tuple, we say that $b$ is at position $k$.
With the above definition, directed edges of the DAG will occur (i) from the receiver to all $b_j$ in $B_1$ ($1\leq j\leq l$), (ii) from each $b_j$ in block $B_b$ to the sender, (iii) moreover, we have edges between nodes in $B_i$ and nodes in $B_{i+1}$.
The following is required:
\begin{enumerate}
  \item If a unique color was to be assigned to each party of the protocol, based on the results of~\cite{DesmedtWangBurmester05}, the sender and receiver can privately communicate, if when choosing any $t$ colours and removing the vertices of the DAG with those $t$ colours the sender and receiver remain connected - meaning that there still exists a directed path from the sender to the receiver on the reduced DAG.
  \item We require that if at level $k$ the parties in $B_{k}$ receive shares of $\pi_i^k$, the parties in $B_{k+1}$ (i.e., at level $k+1$) receive shares of $\pi_i^{k+1}$=$\omega^k_i + \pi_{\rho(i)}^k$.
\end{enumerate}

Two methods can be used to achieve the above requirements. One uses re-sharing - such as the redistribution scheme described in~\cite{DesmedtJajodia97}.
The other uses a large set of MIX servers $X$ to guarantee
the following property.
\begin{definition}\label{def:wireConfinment}
We say that set $X$ of MIX servers is under {\em $t$-confinement} if all members of set $T$ where $|T| = t$ appear in at most $t$ positions over all blocks of MIX servers used and this for all $T \subseteq X$ where $|T| = t$.
\end{definition}
It is easy to see that the above satisfies the DAG requirements.
\subsection{The MIX Protocol}\label{sec:privAnonCommunAbelianProt}
    In the case of Internet voting this is used as a pre-voting protocol for the transmission of voting codes to voters and it is used to achieve anonymity of voting codes.
    We assume $S$ to be a finite Abelian group and denote with $v$ the number of senders, and thus the number of messages (sets of voting codes) that need to be transmitted. In the following, we only describe the required difference when compared to Protocol~\ref{prot:privAnonProtD}.
    \begin{protocol}{Private and Anonymous Random Communication Protocol}\label{prot:privAnonProt2}
    \begin{enumerate}[{\bf {Step }1}]
        \item Let $s_i$ be the $i^{th}$ message (where $1\leq i\leq v$). The sender shares each $s_i$ by choosing $l$ shares $\pi^1_{i,j}\in_R S$ (using an appropriate secret sharing scheme over an Abelian group where $1\leq j\leq l$) and privately sends $\pi^1_{i,j}$ to the corresponding party $B_{1,j}$ in $B_1$.
            \begin{itemize}
            \item As an $(m,b,t)$-verifiers set system is used, $l=t+1$ denotes the number of shares.
            \end{itemize}
        \item\label{enum:PRAMIX1}
            Same as in Protocol~\ref{prot:privAnonProtD}.
        \item\label{enum:PRAMIX12}
            Same as in Protocol~\ref{prot:privAnonProtD}.
        \item
            The \emph{leader} of $B_1$ chooses modifiers $\omega^1_{i,j}\in_R S$ and privately sends $\omega^1_{i,j}$ to parties in $B_1$.
        \item
            Similar as in Protocol~\ref{prot:privAnonProtD}. Only:

        The MIX servers in $B_1$ compute shares of $\pi^2_{i}=\omega^1_{i} + \pi^1_{i}$, i.e. party $P_{j}\in B_{i}$ adds the modifiers it receives from the leader of $B_{i}$ to the share(s) it holds.
        The shares of the $\pi^2_{i}$ are denoted as $\pi^{2}_{i,j}$.
        \item
        If the concept of $t$-confinement is not used, re-sharing of shares $\pi^{2}_{i,j}$ is carried by out by parties in $B_1$ using the redistribution scheme described in~\cite{DesmedtJajodia97}.
        That means that each party in $B_2$ receives $l = t + 1$ values, which they then compress.
        \item
        Steps 2-5 are repeated incrementing by one the indices of $B_1$ and $B_2$ until the last block $B_{b}$ is reached. For all iterations - except when the last block $B_{b}$ is reached, Step 6 is also repeated (except if $t$-confinement is used).
        \item
        If $t$-confinement is not used, shares held by the MIX-servers of block $B_{b}$ are re-shared.
         \item
         Shares held by MIX-servers of block $B_{b}$ are denoted as $\phi_{i,j}$. $MIX_{b,j}\in B_b$ then sends $\phi_{i,j}$ to the $i^{th}$ voter using~\cite{humanPSMTSCN}.
    \end{enumerate}
    \end{protocol}
    It should be noted, that as in~\cite{humanPSMTSCN}, MIX servers will send shares to voters using network disjoint paths, as the communication network cannot be trusted with the adversary capable of listening to at most $t$ of these paths.
    The way voters cast their vote will be described in Section~\ref{sec:evotingProt}.
\subsection{Security Proof}
    In this section we present the security proof for Protocol~\ref{prot:privAnonProt2}.
    \begin{cor}\label{the:secProofAnonProt2}
    Protocol~\ref{prot:privAnonProt2} is a reliable, private and anonymous message transmission protocol.
    \end{cor}
    \begin{proof}
    Formally, we have:\\
    \ \ \textbf{Perfect Reliability} - This is the same as in Theorem~\ref{the:secProofAnonProtD}.

    \noindent \textbf{Perfect Privacy} - The protocol achieves perfect privacy as each message is ``shared'' over $l=t+1$ shares.
    In the case of $t$-confinement, the view of the adversary will consist of at most $t$ shares. This number is one less that the number required to reconstruct a secret and thus perfect privacy is achieved. In the case of re-sharing, the re-sharing guarantees that shares at level $i$ are independent of those at level $i+1$ (note that the adversarial parties are passive). The rest follows from~\cite{DesmedtWangBurmester05} and through the use of re-sharing or $t$-confinement.
    When using re-sharing we ensure that there is no cut of $t$ vertices (colors)
    that can disconnect the sender and the receiver. This is because the resharing of shares makes
    certain that the parties in block $b_i$ receive shares from $t+1$ parties in block $b_{i-1}$.
    So, any adversarial $t$ parties in block $b_{i-1}$ will not allow to cut the graph.
    It is easy to see that the condition of~\cite{DesmedtWangBurmester05} (i.e. no $t$ parties are able to cut a graph) is satisfied when using $t$-confinement thus allowing for secure solutions.

    \noindent \textbf{Perfect Anonymity} - This is very similar to the anonymity proof of Theorem~\ref{the:secProofAnonProtD}. The only difference is that now where a lower number of MIX servers are used, due to Property 3 from the definition of verifier set systems, there exists a block $b_i$ - $1 \leq i \leq b$, free from adversary controlled MIX servers. Because of this, the adversary is unable to learn the modifiers and permutation which are added and implemented respectively to the shares of the messages.
    \makeatletter \let\@xpar=\par\penalty10000\hfill$\Box$\vspace{.5ex}\@xpar\penalty-8000\noindent \makeatother
    \end{proof}
\subsection{Use of non-Abelian Group  - Multi-seat Election Case}
    When a non-Abelian group is used, the protocol is similar to that presented in Section~\ref{sec:privAnonCommunAbelianProt}. Due to the non-Abelian nature of the group, alternative additional techniques will have to be employed to manage the fact that dealing with shares cannot be done locally (due to the multiplication) thus this needs to be shared and securely computed among many parties using techniques presented in~\cite{DesmedtYao12}.

    Suppose we have an election in which we have $s$ seats in which every voter can vote for up to $s$ of the $c$ candidates - where $s\leq c$.
    To enable \emph{blinding} of the code, we give to each voter a secret permutation $\pi \in S_c$, where $S_c$ is the symmetric group.
    For each favourite candidate $i$ the voter wants to vote for, $\pi(i)$ is transmitted to the returning officer.

    Note that $\pi$ is \emph{not} necessarily unique to the election, as opposed to Chaum's code voting. The protocol is organised to avoid that this creates a problem.
    In the case of Internet voting, the following protocol is used as a pre-voting protocol, for the transmission of $v$ number of voting ``codes'' (i.e. permutations) to $v$ number of voters and it is used to achieve anonymity of voting codes. We assume $S=S_c$ to be a finite non-Abelian group.

    It should be noted that the protocol to be presented is only useful for the private and anonymous transmission of permutations with which receivers can cast their vote.
    \begin{protocol}{Private and Anonymous Random Communication Protocol}\label{prot:privAnonProt}
    \begin{enumerate}[{\bf {Step }1}]
        \item Same as in Protocol~\ref{prot:privAnonProt2} only now a non-Abelian group is used and permutations are transmitted.
        \item
            The \emph{leader} of $B_2$ chooses modifiers $\omega^2_{i,j}\in_R S_c^l$ and privately sends $\omega^2_{i,j}$ to parties in $B_2$ such that the $l$ values $\omega^2_{i,j}$ are regarded as shares of $\omega^2_{i}$.\footnote{As shown in~\cite{DesmedtYao12}, to securely compute $\pi$ and $\omega$ where $\pi$ is chosen by one party and $\omega$ by another, we need $2t+1$ parties where $t$ parties are curious. To mimic as closely as possible the working of~\cite{DesmedtYao12}, $\omega^2_{i,j}$ is chosen by the leader of $B_2$ and \emph{not} by the leader of $B_1$.}
        \item
            For each $(i,j)$ the $l$ values $\pi^1_{i,j}$ are regarded as shares of $\pi^1_{i}$.

        The MIX servers in $X'_{1,2} \subseteq X$ where $|X'_{1,2}| \geq 2t+1$ and $B_1 \cup B_2 \subseteq X'_{1,2}$ compute shares of $\pi^{2}_{i}=\omega^2_{i}\circ \pi^1_{i}$ using a black box non-Abelian multiparty computation protocol\footnote{Note that the MIX servers in $B_1 \cup B_2$ can also be a in $X'_{1,2}$ where $|X'_{1,2}| \geq 2t+1$. Additionally, the efficiency of black box non-Abelian multiparty computation protocols is better when $|X'_{1,2}| >> 2t+1$.} (see Section~\ref{sec:otherSubprotsBlackBox}). This is done so that $\omega^2_{i}$ blinds $\pi^1_{i}$. The shares of the product are denoted as $\pi^{2}_{i,j}$ and are obtained by the parties\footnote{Note that~\cite{DesmedtYao12} allows to organise the computation such that the output, i.e. shares of $\pi^{2}_{i}$, are received by parties in $B_2$.} in $B_2$.

        \item
        The \emph{leader} of $B_2$ informs all other MIX servers in $B_2$ how they have to permute the $i$-index of all shares they hold from the above operations. This permutation is defined by $\rho_2\in_R S_v$. On the $i$ indices the MIX servers in $B_2$ apply the permutation $\rho_2$. So, $\pi^{2}_{i,j}:=\pi^2_{\rho_2(i),j}$.
        \item
        The above three steps are repeated by incrementing by one the indices of $B_1$ and $B_2$ (thus $B_k \neq B_{k+1}$).
        After parties in $B_k$ permute the $i$ indices of $\pi^{k}_{i,j}$ using $\rho_k$ - where $2 \leq k \leq b-1$, the \emph{leader} of $B_{k+1}$ chooses modifiers $\omega^3_{i,j}\in_R S_c^l$  which are given to parties in $B_k$, the black box non-Abelian multiparty computation sub-protocol is executed by parties in $X'_{k,k+1} \subseteq X$ where $B_k \cup B_{k+1} \subseteq X'_{k,k+1}$  $|X'_{k,k+1}| \geq 2t+1$ and the process continues till the final block of servers $B_b$ is reached.
        \item
       After parties in $B_b$ permute the $i$ indices of $\pi^{b}_{i,j}$ using $\rho_b$, the \emph{leader} of $B_1$ chooses modifiers $\omega^1_{i,j}\in_R S_c^l$  which are given to parties in $B_1$, the black box non-Abelian multiparty computation sub-protocol is executed between parties in block $B_b$ and $B_1$ and the output of which is held by parties in $B_1$. $MIX_{1,j}\in B_1$ sends the output it holds to the $i^{th}$ voter using~\cite{humanPSMTSCN}.
    \end{enumerate}
    \end{protocol}
    It should be noted, that as in~\cite{humanPSMTSCN}, MIX servers will send shares to voters using network disjoint paths, as the communication network cannot be trusted with the adversary capable of listening to at most $t$ of these paths.
    The way voters will use what they receive to cast their vote will be described in Section~\ref{sec:evotingProt}.

    We now present the security proof for Protocol~\ref{prot:privAnonProt}.
    \begin{theorem}\label{the:secProofAnonProt}
    Provided Protocol~\ref{prot:privAnonProt} together with the appropriate black box non-Abelian multiparty computation sub-protocol is used, then Protocol~\ref{prot:privAnonProt} is a reliable, private and anonymous random transmission protocol.
    \end{theorem}
    The proof of the above theorem is similar to the proof of Theorem~\ref{the:secProofAnonProtD}, but
      relying on either~\cite{CohenDIKMRR13,DesmedtYao12}.
\section{Electronic Code Voting Protocol}\label{sec:evotingProt}
    In this section we outline how components of previous sections are combined.
\subsection{Preparation, Mixing and Transmission of Voting Codes}
    As described in Section~\ref{sec:description} the CGE is responsible for creating the codes with which voters will cast their votes.
    We first explain this for the single-seat election.

    Considering an election has $c$ number of candidates and that there are $v$ number of voters, the CGE will create $v$ random \emph{initial} codes for each of the $c$ candidates.
    In total, $c \times v$ unique number of codes will be generated. The CGE will then group these codes to form $v$ number of $c-tuples$, with each tuple containing a single
    code for each of the $c$ candidates.

    Each of these codes will then be transmitted as one-time pads to the voters in the same way as described by Protocol~\ref{prot:privAnonProt2}.
    It should be noted that Protocol~\ref{prot:privAnonProt2} describes the transmission of only $v$ codes as opposed to
    $c \times v$ required by the voting protocol.
    To transmit all the voting codes, $c$ executions of Protocol~\ref{prot:privAnonProt2} will be executed at the same time.
    These executions should \emph{not be independent between them but instead should use the same permutations} ($\rho \in_R S_{v}$ in Step 2) and modifiers ($\omega_{i,j}$ in Step 4) used throughout all executions of the protocol, i.e. the same modifier is used for all codes the same voters will receive and they remain bundled together (i.e. by reusing $\rho$). These $c$ executions can be carried out either in parallel or sequentially, as long as each voter receives $c$ voting codes.

    In the case of multi-seat elections, each voter will receive a single permutation over $S_c$ - which is a permutation of the alphabetical ordering of the candidates. Moreover, Protocol~\ref{prot:privAnonProt} will be used.
\subsection{Receiving and Reconstructing Voting Codes}
    We first explain the single-seat case.
    Each voter will receive $l=t+1$ shares for each voting code, receiving each one using a \emph{different} computational device. It should be noted that the $i^{th}$ share of each of the $c$ voting codes will be received upon the \emph{same} computational device.

    The voter can then identify the code which corresponds to the candidate of their choice. Once all pieces of each code are received, the code corresponding to their choice can be reconstructed in a similar manner as described in Section~\ref{sec:otherSubprotsHumanPSMT}.

    In the multi-seat election, instead of receiving a $c$-tuple, a single permutation is received - which is a permutation of the alphabetical ordering of the candidates. Similar to the single seat case, $t+1$ shares of this permutation will be received by the voter who will reconstruct the permutation as described in~\cite[Section~4.2, Section~4.3]{humanPSMTSCN}. This will allow the voter to identify the candidates of their choice. Supposing the voter wants to vote for candidate $c$ and candidate $c'$, the reconstruction of the permutation will help the voter identify $\pi(c)$ and $\pi(c')$ which correspond to the candidates of their choice. To cast their vote, voters will have to send back to the CGE these $\pi(c)$ and $\pi(c')$ values.
\subsection{Transmission, Mixing and Counting of Cast Votes}\label{sec:sendinVotes}
    We first explain this for the single-seat case.
    A voter identifies the code corresponding to the candidate of their choice and sends this code back to the CGE by transmitting this code \emph{to the leader} of the last block of MIX.

    To transmit voter codes in the reverse direction (towards the CGE), \emph{the leaders} of each block of MIX servers will have to carry out the reserve operations on the codes. Thus the inverse permutations ($\rho_{b}^{-1}$) and modifiers ($-\omega_{i}^k$) are used.
    Once a code arrives to the CGE, it will identify the candidate it corresponds to and the vote will be counted.

    The multi-seat case is similar. Once a voter identifies one of the $\pi(c)$ which corresponds to one of their chosen candidates, they will have to send this $\pi(c)$ to the leader of the last block of MIX servers. Similar to the single-seat case, the reserve operations on the codes will have to be carried out
    Once a voter's $\pi(c)$ arrives to the CGE, the CGE will apply $\pi^{-1}$ and identify the candidate the voting corresponds to and the vote will be counted.

    \noindent \textbf{Acknowledgements:} The authors would like to thank the anonymous referees for their
    valuable comments on improving the presentation and clarity of this paper.

    \noindent The authors would also like to thank Rebecca Wright, Juan Garay and Amos Beimel for expressing their interests in Private and Secure Message Transmission in which one cannot trust the equipment used by the receiver.
    \newpage
\bibliographystyle{abbrv}
\bibliography{voteID2015}
\end{document}